\documentclass[runningheads]{llncs}
\usepackage{graphicx}
\usepackage{epsfig}
\usepackage{graphics}
\usepackage{array}
\usepackage{comment}
\usepackage{fullpage}
\date{}

\title{A Characterization of all Stable Minimal Separator Graphs} 
\author{Mrinal Kumar \and Gaurav Maheswari \and N.Sadagopan} 
\institute{Department of Computer Science and Engineering, Indian Institute of Technology, Chennai-600036, India.\\
\email{\{mrinalk,gauravm,sadagopu\}@cse.iitm.ac.in}}
\begin{document}
\maketitle
\begin{abstract}
In this paper, our goal is to characterize two graph classes based on the properties of minimal vertex (edge) separators.   We first present a structural characterization of graphs in which every minimal vertex separator is a stable set.  We show that such graphs are precisely those in which the induced subgraph, namely, a cycle with exactly one chord is forbidden.  We also show that deciding maximum such forbidden subgraph is NP-complete by establishing a polynomial time reduction from maximum induced cycle problem \cite{garey}.  This result is of independent interest and can be used in other combinatorial problems.  Secondly, we prove that a graph has the following property: every minimal edge separator induces a matching (that is no two edges share a vertex in common) if and only if it is a tree. 

\end{abstract}
\section{Introduction}
Vertex connectivity is a classical topic in graph theory and motivated many problems in structural graph theory.   One such problem is the study of constrained vertex separators.  In particular, clique separators, stable separators, and balanced separators are the most popularly studied constrained vertex separator in the literature \cite{dirac,clique,stable,marx}.  This line of study was initiated by Dirac \cite{dirac} with a structural characterization of all clique vertex separators graphs.  More precisely, Dirac addressed a fundamental question of characterizing graphs in which every minimal vertex separator is a clique.\\
A graph is chordal (also known as triangulated), if there is no induced cycle of length at least 4.  Alternately, every cycle of length at least 4 has a chord.   In \cite{dirac} Dirac presented a structural characterization of chordal graphs with respect to vertex separators.  The famous characterization says, a graph is chordal if and only if every minimal vertex separator is a clique.   While chordal graphs and its structural properties have received much attention in the literature, the analogous question, characterize graphs such that every minimal vertex separator is an independent set has not received attention in the past.    \\
Another motivation for the study of stable minimal vertex separator graphs is from the theory of {\em contractible edges}.   In a $k$-connected graph, an edge is said to be contractible if its contraction does not decrease the connectivity below $k$.  This problem and its structural study was initiated by Tutte \cite{tutte}.  Subsequently, Saito \cite{saito} characterized contractible edges in 3-connected graphs and in \cite{saito} it is shown that there exists a contractible edge in every 3-connected graphs.  Unlike 3-connected graphs, there are 4-connected graphs without any contractible edge and these graphs are called {\em critically} 4-connected graphs. Martinov in \cite{martinov} characterized critically 4-connected graphs.  A natural question is to characterize critically $k$-connected graphs for any $k$.  The following lemma relates vertex connectivity and edge contraction.
\begin{lemma}\cite{saito}
Let $G$ be a $k$-connected graph and $G.e$ denote the graph obtained from $G$ by contracting the edge $e=\{u,v\}$.  $G.e$ is $k$-connected iff for any minimum vertex separator $S$ of $G$, $\{u,v\} \not\subseteq S$.  
\end{lemma}
In the light of the above lemma, we pose a natural question, which is to characterize graphs such that every edge of it is contractible.  Equivalently, characterize graphs in which every minimum vertex separator is a stable set.  However, from the point of view of structural graph theory, it is appropriate to characterize graphs in which every minimal vertex separator is a stable set.   With these two motivations the important contribution in this paper are the following:
\begin{itemize}
\item Graph $G$ does not have a cycle of length at least 4 with exactly one chord as an induced subgraph (1-chord subgraph) if and only if every minimal vertex separator in $G$ is a stable set.  \\
\item We establish the fact that deciding whether a graph has a maximum 1-chord subgraph is NP-complete. \\
\item We have also looked at an analogous question in the edge connectivity setting.  We show that the class of graphs in which every minimal edge separator induces a matching are precisely the class of trees.
\end{itemize}
\section{Graph Preliminaries}
\label{prelims}
Notation and definitions are as per \cite{west,golu}.  Let $G =(V,E)$ be an undirected non weighted simple graph, where $V(G)$ is the set of vertices and $E(G) \subseteq \{\{u,v\}~|~ u,v \in V(G)$, $u \not= v \}$.  A separating set or a vertex separator of a graph $G$ is a set $S \subseteq V(G)$ such that the induced subgraph, denoted by $G \setminus S$, on the vertex set $V(G) \setminus S$ has more than one connected component. The vertex connectivity of a graph $G$, written  $\kappa(G)$, is the minimum cardinality of a vertex set $S$ (minimum vertex separator) such that $G \setminus S$ is disconnected or has only one vertex. A graph is $k$-connected if its vertex connectivity is $k$. A vertex separator  $S$ is called a $a$-$b$ vertex separator iff $S$ disconnects $a$ and $b$. i.e. the graph $G \setminus S$ has at least two connected components such that $a$ and $b$ are in distinct components.  A $a$-$b$ vertex separator is said to be a minimal vertex separator iff no proper subset of it is a $a$-$b$ vertex separator.  A minimum $a$-$b$ vertex separator is a minimal $a$-$b$ vertex separator of least size.  For $S \subset V(G)$, $G[S]$ denote the graph induced on the set $S$.  For a $a$-$b$ vertex separator $S$ let $\{C_1,\ldots,C_r\}$ denote the set of connected components in $G \setminus S$.  A chord of a cycle $C$ is an edge joining a pair of nonadjacent vertices in $C$.  A graph is said to be chordal if every cycle of length at least 4 has a chord in it.  A subgraph $H$ of $G$ is said to be {\em 1-chord} if $G[V(H)]$ is a cycle of length at least 4 with exactly one chord.  $G$ is {\em 1-chord free} if $G$ contains no such $H$.  An example is shown in Figure \ref{1chordfig}.
\begin{figure}
\begin{center}
\includegraphics[scale=0.6]{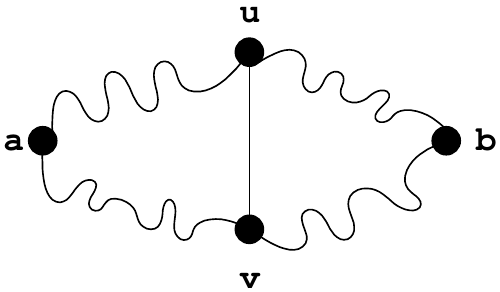}
\caption{1-chord graph}\label{1chordfig}
\end{center}
\end{figure}
\section{Stable Vertex Separator graphs: Structural Characterization}
In this section, we characterize graphs in which all minimal $(a,b)$ vertex separators are stable sets.   For simplicity, we use minimal vertex separator instead of minimal $(a,b)$ vertex separator and the pair $(a,b)$ will be clear from the context.  The following well known lemma is a key lemma in the proof of our main theorem and we describe a proof of the same to make the paper self contained.
\begin{lemma}
\label{char-lemma}
Let $S$ be a $(a,b)$ vertex separator in $G$ and $C_1$ and $C_2$ are connected components in $G \setminus S$ such that $a \in C_1$ and $b \in C_2$.  $S$ is minimal iff for each $v \in S$ there exists $u \in C_i, i \in \{1,2\}$ such that $\{u,v\} \in E(G)$.  
\end{lemma}
\begin{proof}
Suppose there exists $v \in S$ for all $u \in C_1$  there is no edge $\{u,v\}$ in $G$.  Consider the set $S'=S \setminus \{v\}$.  Clearly in $G \setminus S'$ the pair $\{a,b\}$ are in distinct components.  This implies that $S'$ is a $(a,b)$ vertex separator.   Since $S' \subset S$, implies that $S$ is not a minimal $(a,b)$ vertex separator.  This contradicts the given fact that $S$ is minimal $(a,b)$ vertex separator.  Therefore, for each $v \in S$ there exists $u \in C_1$ such that $\{u,v\} \in E(G)$.  A similar argument proves that for each $v \in S$ there exists $u \in C_2$ such that $\{u,v\} \in E(G)$.   Conversely,  assume there exists $S' \subset S$ in $G$ such that $S'$ is a minimal $(a,b)$ vertex separator.   Let $C'_1$ and $C'_2$ are the connected components in $G \setminus S'$ such that $a \in C'_1$ and $b \in C'_2$.   This implies that for each $z \in C'_1$ and for each $z' \in C'_2$ every path $P_{zz'}$ from $z$ to $z'$ contains an element from $S'$.  Since $S' \subset S$ there must exist $w \in S \setminus S'$ such that $\{w,w'\} \notin E(G)$ for any $w' \in C_1$ or $\{w,w''\} \notin E(G)$ for any $w'' \in C_2$.   However, this contradicts the given hypothesis that  for each $v \in S$ there exists $u \in C_i, i \in \{1,2\}$ such that $\{u,v\} \in E(G)$.   Therefore, the sufficiency follows.  Hence the lemma.  \qed
\end{proof}
\begin{theorem}
\label{1chordchar}
$G$ is 1-chord free if and only if every minimal $(a,b)$ vertex separator of $G$ is an independent set. 
\end{theorem}
\begin{proof}
{\em Necessity:} Given that $G$ is 1-chord free we now show that every minimal $(a,b)$ vertex separator is an independent set.  We present a proof by contradiction.  Suppose there exists a minimal $(a,b)$ vertex separator $S$ such that $G[S]$ is not an independent set.  This implies that there exists $u,v \in S$ such that $\{u,v\} \in E(G)$.  Consider the connected components $\{C_1,\ldots,C_r\}$ in $G \setminus S$.  Without loss of generality we assume that $a \in C_1$ and $b \in C_2$.  Since $S$ is a minimal $(a,b)$ vertex separator, by Lemma \ref{char-lemma}, for each $w \in S$ there exists $z \in C_1$ such that $\{w,z\} \in E(G)$.  Similarly there exists $z' \in C_2$ such that $\{w,z'\} \in E(G)$.  Consider shortest paths $P_1=\{u,z_1,\ldots,z_k,v\}, k \geq 1, z_i \in C_1$ and $P_2=\{u,w_1,\ldots,w_l,v\}, l \geq 1, w_i \in C_2$.  Since $P_1$ and $P_2$ are shortest paths between $u$ and $v$ and $\{u,v\} \in E(G)$ implies that $P_1$ and $P_2$ together form a cycle of length at least 4 with exactly one chord.  In other words, $G$ contains an induced 1-chord graph as a subgraph.  However, this is a contradiction to the given hypothesis.  Therefore, our assumption that there exists a minimal $(a,b)$ vertex separator $S$ such that $G[S]$ is not an independent set is wrong.  Hence the necessity follows. \\
{\em Sufficiency:} We prove that $G$ is 1-chord free by the method of contradiction.  Assume that in $G$ there exists an induced 1-chord graph $H$ as a subgraph.  We now construct a minimal $(a,b)$ vertex separator $S$ such that $G[S]$ is not an independent set.  Let $V(H)=\{u,a,z_1,\ldots,z_l=v,y_1,\ldots,y_k=b\}, l \geq 1, k \geq 1$ and $E(H)=\{\{u,a\},\{a,z_1\},\{v,y_1\},\{b,u\},\{u,v\}\} \cup  \{\{z_i,z_{i+1}\}, 1 \leq i \leq l-1 \} \cup \{ \{y_j,y_{j+1}\}, 1 \leq j \leq k-1\}$.  Let $X=\{a\} \cup \{z_i ~|~ z_i \in V(H)\} \setminus \{v\}$ and $Y=\{b\} \cup \{y_j ~|~ y_j \in V(H)\}$.  Note that any minimal vertex separator $S$ in $G$ separating $X$ and $Y$ must contain $u$ and $v$.   Moreover, we observe that $S$ is also a minimal $(a,b)$ vertex separator in $G$.  The reason this is true is due to the following: clearly, $u \in S$ and if $ v \notin S$ then there exists a path between $a$ and $b$ through $v$, contradicting the fact that $S$ is a vertex separator in $G$.  Therefore, $S$ is a minimal $(a,b)$ vertex separator such that $\{u,v\} \subset S$.  Since $\{u,v\} \in E(G)$ implies that $G[S]$ is not an independent set. A contradiction to the hypothesis.  Therefore, $G$ is 1-chord free.  Hence the theorem.  \qed 
\end{proof}
We now present two more combinatorial observations on 1-chord free graphs with respect to its vertex connectivity.
\begin{lemma}
\label{2connectobs}
Let $G$ be a non complete at least 2-connected graph.  If $G$ is 1-chord free then $G$ is triangle free.
\end{lemma}
\begin{proof}
Suppose $G$ is not triangle free.  Let $\{a,b,c\}$ induce a triangle in $G$.  Since $G$ is at least 2-connected and non complete, there must exist a path between $b$ and $c$ avoiding $a$.  Let $P_{bc}$ denote a shortest such path and $V(P_{bc})=\{b,z_1,\ldots,z_i,c\}, i \geq 1$.   We now show that $G$ contains 1-chord graph as a subgraph by considering three cases.  An illustration is given in Figure \ref{case}.   {\em Case 1:} There is no edge $\{a,z_j\}, 1 \leq j \leq i$, for any $z_j \in V(P_{bc})$.  Clearly, $\{a,b,c\}$ together with $P_{bc}$ induce a 1-chord subgraph in $G$.  However, we know that $G$ is 1-chord free.  A contradiction.  Therefore, $G$ is triangle free.  Note that for other two cases, $i \geq 2$ as $G$ is a non complete graph.  {\em Case 2:} There exists an edge $\{a,z_j\}, 1 \leq j \leq i-1$, for some $z_j \in V(P_{bc})$.   If $a$ is adjacent more than one $z_j$ then without loss of generality we choose the $z_j$ such that $j$ is the least.  Now, the set $\{c,a,b,z_1,\ldots,z_j\}$ induces a 1-chord subgraph with $\{a,b\}$ as the chord. A contradiction in this case too. {\em Case 3:} $\{a,z_j\} \in E(G), z_j=z_i$.   In this case the set $\{b,a,c,z_i\}$ is a 1-chord subgraph with $\{a,c\}$ as the chord.  A contradiction.  This completes our case analysis and we see that our assumption that $G$ is not triangle free is wrong.  Therefore, the claim follows. \qed
\end{proof}
\begin{figure}
\begin{center}
\includegraphics[scale=0.75]{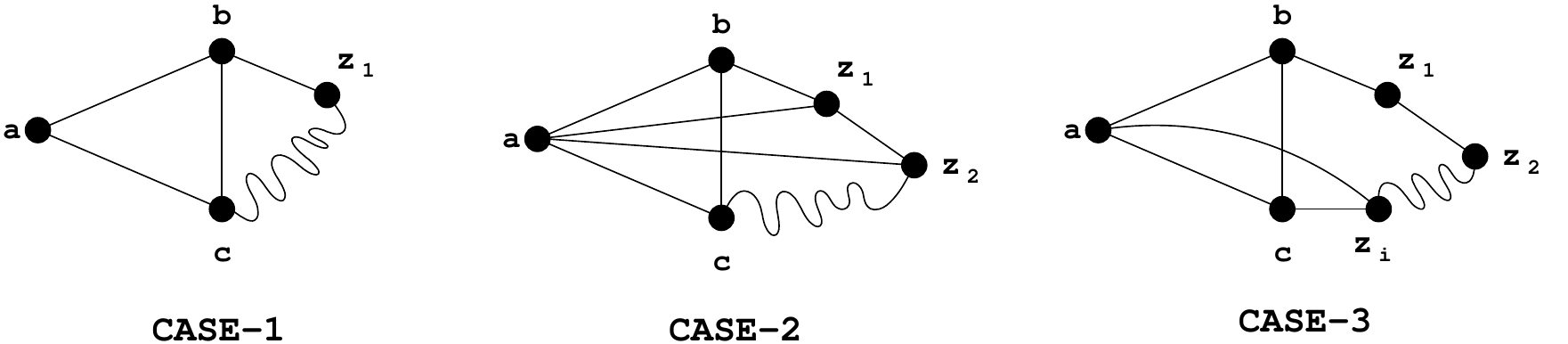}
\caption{An illustration for the proof of Lemma \ref{2connectobs}}\label{case}
\end{center}
\end{figure}
Note that the converse of the above lemma is not true.  For example, trees are triangle free and 1-chord free, however, trees are not 2-connected.  The following observation characterizes the connected components in 1-connected 1-chord free graphs.
\begin{lemma}
Let $G$ be an exactly 1-connected 1-chord free graph.   For a cut-vertex $v$ of $G$, let $\{C_1,\ldots,C_r\}$ denote the connected components in $G \setminus \{v\}$.  For each cut-vertex $v$, the subgraph induced on $V(C_i) \cup \{v\}$ is either a complete graph or a 1-chord free graph.
\end{lemma}
\begin{proof}
The claim follows from the contrapositive of Lemma \ref{2connectobs} and the fact that 1-chord freeness is a hereditary property. \qed
\end{proof}
The characterization in Theorem \ref{1chordchar} can be used to test whether a given graph has the property that every minimal vertex separator is an independent set.  In particular, this calls for testing the existence of 1-chord subgraph in the given graph.  Two of the closely related problems are finding minimum (maximum) sized 1-chord subgraph of a graph.  In the next section, we show that decision version of maximum sized 1-chord subgraph is NP-complete by presenting a polynomial time reduction from Maximum induced cycle problem.  This result is of independent interest and can be used in other combinatorial problems.  The decision version of maximum 1-chord subgraph is given below.
\begin{center} {\em Maximum 1-chord subgraph problem}
\begin{tabular}{|p{14cm}|}
\hline \\ 
{\bf Instance:} Graph $G$, and an integer $l$\\
{\bf Question:} Is there a subgraph $H$ of $G$ such that $|V(H)| \geq l$ and $G[V(H)]$ is a 1-chord? \\
\hline 
\end{tabular}
\end{center} 
\begin{theorem}
\label{1chordredn}
Decision version of maximum 1-chord subgraph of a graph is NP-complete
\end{theorem}
\begin{proof}
{\bf 1-chord subgraph is in NP:} Given a certificate $C=(G,H,l)$, to witness the fact that this problem is in NP, we now present a deterministic polynomial time algorithm to verify the validity of $C$.  Observe that 1-chord graph of size $l$ has the degree sequence $(3,3,2,\ldots,2)$ with exactly $l-2$ vertices of degree 2.  Also, both degree 3 vertices are adjacent and deleting the corresponding edge between them results in a cycle.  It is now clear that the above two crucial observations along with standard Depth First Search algorithm can verify whether $C$ is valid or not in time polynomial in the input size.  Therefore, we conclude that 1-chord subgraph is in NP. \\
{\bf 1-chord subgraph is NP-hard:} We establish a polynomial time reduction from maximum induced cycle problem and its decision version is known to be NP-complete\cite{garey}.  The decision version of the problem is described as follows:
\begin{center}{\em Maximum induced cycle problem} 
\begin{tabular}{|p{14cm}|}
\hline \\ 
{\bf Instance:} Graph $G$, and an integer $l$\\
{\bf Question:} Is there a subgraph $H$ such that $|V(H)| \geq l$ and $G[V(H)]$ is a cycle ? \\
\hline 
\end{tabular}
\end{center} 
Given an instance $(G,H,k)$ of induced cycle problem, we construct an instance $(G',H',2k)$ of 1-chord subgraph as follows: $V(G')=V(G) ~ \cup ~\{ \{ v_{ij}^1,v_{ij}^2,\ldots,v_{ij}^k\} ~|~ \{v_i,v_j\} \in E(G)\}$ and $E(G')=E(G) ~ \cup ~ \{\{v_{ij}^p,v_{ij}^{p+1}\} ~|~ 1 \leq p \leq k-1\} ~\cup~ \{\{v_i,v_{ij}^1\}, \{v_{ij}^k,v_j\}\}$.  An example is illustrated in Figure \ref{1chordredn}.  We now show that $(G,H,k)$ has an induced cycle of size at least $k$ if and only if $(G',H',2k)$ has a 1-chord subgraph of size at least $2k$.  For {\em only if} claim, $G$ contains an induced cycle $H$, $|V(H)| \geq k$.  By our construction of $G'$, for any edge $e=\{v_i,v_j\} \in E(H)$ there is a path $P_{v_iv_j}$ using the vertex set $\{ v_{ij}^1,v_{ij}^2,\ldots,v_{ij}^k\}$.  Clearly, in $G'$, $V(H)$ together with $\{ v_{ij}^1,v_{ij}^2,\ldots,v_{ij}^k\}$ induce a 1-chord subgraph with $e$ as the unique chord.   Thus, we have constructed  in $G'$, a 1-chord subgraph $H'$ such that $|V(H')| \geq 2k$.  For {\em if} claim, $G'$ contains 1-chord subgraph $H'$ of size at least $2k$.  Note that if such a $H'$ exists in $G'$ then either $V(H') \cap \{ v_{ij}^1,v_{ij}^2,\ldots,v_{ij}^k\} = \phi$ or  $\{ v_{ij}^1,v_{ij}^2,\ldots,v_{ij}^k\} \subset V(H')$.  Also, $H'$ has two induced cycles with at least one is of size at least $k$.  This implies that, in either case,  there exists an induced cycle $H$ in $G$ such that $|V(H)| \geq k$.  Hence the claim.  Note that $|V(G')|=|V(G)|+k|E(G)|$ and $|E(G')|=(k+2)|E(G)|$ and $G'$ can be constructed in $O(k|E(G)|)$ time.  For non trivial cases, $k$ will always be less than or equal to $|E(G)|$, so time complexity will be $O(|E(G)|^2)$.  Thus, we have established a polynomial time reduction from induced cycle problem to 1-chord subgraph problem.  As a consequence, we conclude that deciding 1-chord subgraph is NP-hard.  Therefore, 1-chord subgraph problem is NP-complete.  Hence the theorem. \qed
\begin{figure}
\begin{center}
\includegraphics[scale=0.75]{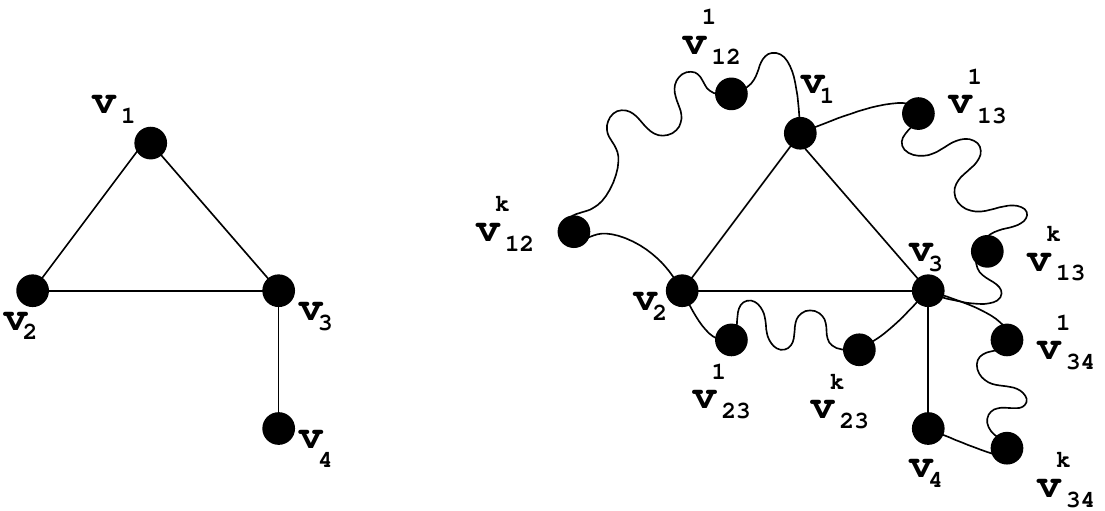}
\caption{Reducing an instance of induced cycle problem to 1-chord graph problem}\label{1chordredn}
\end{center}
\end{figure}
\end{proof}
{\bf Remark:} It is interesting to compare our NP-hardness result with an open problem posed in \cite{spinrad-open} and it is the following: can we determine the longest cycle without crossing chords in polynomial time?  Note that maximum induced cycle problem and maximum 1-chord subgraph are special cases of this problem and both are known to be NP-complete.  With these observations we believe that the above problem may not have a polynomial time algorithm.\\
{\bf Other Observations:} As far as approximation algorithm is concerned for the above problem, there is no polynomial time approximation algorithm with approximation ratio $O(n^{1-\epsilon})$ for any $\epsilon \geq 0$.  This is true due to the following observation.  NP-hard reduction of induced cycle problem is from independent set problem and this reduction is an approximation ratio preserving reduction.  Since independent set does not have an approximation algorithm with approximation ratio $O(n^{1-\epsilon})$ for any $\epsilon \geq 0$, it follows that induced cycle problem does not have an approximation algorithm with approximation ratio $O(n^{1-\epsilon})$ for any $\epsilon \geq 0$.  Moreover, NP-hard reduction of Theorem \ref{1chordredn} is also an approximation ratio preserving reduction and hence we conclude that 1-chord subgraph problem does not have an approximation algorithm with approximation ratio $O(n^{1-\epsilon})$ for any $\epsilon \geq 0$.  As far as parameterized complexity is concerned with parameter as the size of 1-chord subgraph, we observe that parameterized 1-chord subgraph is $W[1]$-hard. i.e. there is no parameterized algorithm with parameter as the size of 1-chord subgraph.  More about parameterized complexity can be found in the book by Downey and Fellows \cite{downey}. The above result follows from the fact that parameterized independent set is $W[1]$-hard and NP-hard reduction of independent set to induced cycle and NP-hard reduction of Theorem \ref{1chordredn} are parameterized reduction.  Therefore, parameterized 1-chord subgraph is $W[1]$-hard. 
\section{Structural Characterization of Matching Edge Separators}
We shall now focus on an analogous question, which is to characterize the graph class such that any minimal edge separator induces a matching.  An edge separator of a graph $G$ is a set $E' \subseteq E(G)$ such that the induced subgraph $G'$, $V(G')=V(G)$ and $E(G')=E(G)\setminus E'$ has two connected components.  We denote the induced subgraph $G'$ by $G \setminus E'$.  An edge separator $E'$ is called a $a$-$b$ edge separator if $E'$ disconnects $a$ and $b$. i.e. the graph $G \setminus E'$ has two connected components such that $a$ and $b$ are in distinct components.  A $a$-$b$ edge separator is said to be a minimal $a$-$b$ edge separator iff no proper subset of it is a $a$-$b$ edge separator.  A minimum $a$-$b$ edge separator is a minimal $a$-$b$ edge separator of least size.  
\begin{theorem}
A graph $G$ is such that for all $a,b \in V(G)$, all minimal $a$-$b$ edge separators induce a matching if and only if $G$ is a tree
\end{theorem}
\begin{proof}
For {\em if claim}, since $G$ is a tree, it is a well known fact that between any pair $(a,b)$ of vertices there exists exactly one path between $a$ and $b$.  This implies that any minimal $a$-$b$ edge separator contains exactly one edge, which is a matching.  Hence the claim.  For {\em only if claim},  we present a proof by the method of contradiction. Suppose $G$ is not a tree.  This implies there exists a cycle $C$ in $G$. Now let us consider two vertices $a$ and $b$ in $C$ such that $\{a,b\} \notin E(G)$.  It is a well known fact that for any two vertices in $C$ there exists two edge disjoint paths between them.  In particular, this observation is true for $a$ and $b$. Let ${\cal P}=\{ P_{ab} ~|~ P_{ab}$ is a path between $a$ and $b$ not containing $a$ as an internal vertex $\}$ and $X= \{x ~|~ x \in V(P_{ab}), P_{ab} \in {\cal P}\}$.  In other words, the set $X$ is the set of vertices which lie on at least one path $P_{ab}$.  Consider the set $Y=N_G(a) \cap X$, that is $Y$ is set of all vertices which are adjacent to $a$ and lie on a path from $a$ to $b$ which does not contain $a$ as an internal vertex. Now we will prove that the set of edges $E'=\{\{a,x\}~|~ x \in Y \} $ form a minimal $a$-$b$ edge separator in $G$.  Since $a$ and $b$ lie on $C$ it is clear that $E'$ must contain the edges $\{a,y\}$ and $\{a,z\}$ where $y,z \in V(C)$.  This is true due to the fact that $\{a,y\}$ and $\{a,z\}$ lie on two edge disjoint paths between $a$ and $b$.   What follows is that $|E'| \geq 2$ and $E'$ is not
a matching, because any two edges in $E'$ has $a$ as the common vertex.  We now show that $E'$ is a $a$-$b$ edge separator.  Suppose $E'$ is not a $a$-$b$ edge separator, then there is a path $P'_{ab}$ from $a$ to $b$ in $G$ not containing any edge from $E'$.  Without loss of generality we assume that $P'_{ab}$ does not contain $a$ as an internal vertex.  If $P'_{ab}$ contains $a$ as an internal vertex
then we can remove all the vertices until the last $a$ in the path and still get a path from $a$ to $b$.  Let $u$ be the vertex adjacent
to $a$ in $P'_{ab}$, then clearly $u \in X$ and since $\{a,u\} \in E(G)$,  $u \in Y$ as well.  This implies $\{a,u\} \in E'$, which is a contradiction to our assumption that $P'_{ab}$ does not contain any edges from $E'$.  Therefore, $E'$ is a $a$-$b$ edge separator.  Now to prove the minimality of $E'$, suppose $E'$ is not minimal then there is at least one edge $e=\{a,x\} \in E'$ such that
$E' \setminus \{e\} $ is still a $a$-$b$ edge separator in $G$.  Since $e \in E'$, we know that $x \in X$ and there is at least one path $P_{ab} \in {\cal P}$, containing $x$ as an internal vertex. Now consider the sub path $P'$ of $P_{ab}$ from $x$ to $b$. It is clear that $a\notin P'$ as $P_{ab}$ by definition does not contain $a$ as an internal vertex.  So there does not exist $y \in V(G)$ such that $\{a,y\} \in E(G)$ and $\{a,y\}$ is an edge of the path $P'$.   In what follows, the edge $\{a,x\}$ together with the path $P'$ yields a path from $a$ to $b$, which does not contain any edge from $E' \setminus \{e\}$, giving rise to a contradiction that $E \setminus \{e\} $ is a $a$-$b$ edge separator in $G$. Hence we conclude that $E'$ is a minimal $a$-$b$ edge separator which is not a matching.  A contradiction to the hypothesis.  Therefore our assumption that $G$ is not a tree is wrong.  Hence the theorem. \qed
\end{proof}
\bibliographystyle{splncs}
\bibliography{references}
\end{document}